\documentclass[copyright,creativecommons]{eptcs}
\providecommand{\event}{AiML 2026} 

\usepackage{aiml26}

  \usepackage{underscore}         
  \usepackage[T1]{fontenc}        

\usepackage{amsmath}
\usepackage{amssymb,amsthm}
\usepackage{xcolor}
\usepackage{mfl}
\usepackage{enumitem,mathtools}
\usepackage{bbold}
\usepackage{orcidlink}
\usepackage{multicol}
\usepackage{comment}

\renewcommand{\alg}[1]{{\ensuremath{\boldsymbol {\mathbf{#1}}}}}
\renewcommand{\R}{\ensuremath{\alg{R}}}
\newcommand{\rad}{\ensuremath{\mathsf{rad}}}
\newcommand{\Hom}{\ensuremath{\mathrm{Hom}}}
\newcommand{\A}{\alg{A}}

\newcommand{\mAb}{\ensuremath{\mathsf{mAb}}}

\newcommand{\pmAb}{\ensuremath{\mAb_{\mathsf{p}}^-}}
\newcommand{\Lang}{\mathfrak{L}}
\newcommand{\Prop}{\mathit{Prop}}
\newcommand{\Val}{\mathcal{V}}

\title{Pointed Modal Abelian Logic, Algebraically}
\author{Filip Jankovec \orcidlink{0009-0002-2746-0042}  
\institute{Institute of Computer Science of the 
Czech Academy of Sciences \\
Department of Algebra, Faculty of Mathematics and Physics, Charles University \\
Prague, The Czech Republic}
\email{ {jankovec@cs.cas.cz}}
\and
Wolfgang Poiger \orcidlink{0009-0002-8485-5905}
\institute{Institute of Computer Science  
of the Czech Academy of Sciences \\
Prague, The Czech Republic.}
\email{poiger@cs.cas.cz}
}

\newcommand{\titlerunning}{Pointed Modal Abelian Logic, Algebraically}
\newcommand{\authorrunning}{F. Jankovec and W. Poiger}

\hypersetup{
  bookmarksnumbered,
  pdftitle    = {\titlerunning},
  pdfauthor   = {\authorrunning},
  pdfsubject  = {EPTCS},               
  pdfkeywords = {Modal Abelian Logic, Abelian l-groups, Real-valued Kripke Semantics, Truth Lemma, Algebraic Completeness} 
}

\begin{document}
\maketitle

\begin{abstract}
In this article, we investigate the pointed modal logic of reals. We first establish its relational (Kripke) semantics with bounded valuations in the Abelian $\ell$-group of real numbers with the distinguished negative constant $-1$. To study this logic algebraically, we introduce the variety of negatively pointed modal Abelian $\ell$-groups, in particular we focus on the strongly pointed members thereof. Constructing complex algebras and canonical frames, we establish a Truth Lemma connecting the relational and algebraic frameworks. 
Since finitary axiomatizations cannot fully capture the Kripke validities of the reals we introduce further algebraic constraints, in particular including an infinitary Archimedean-style rule. 
Finally, we prove a corresponding `infinitary algebraic completeness' result for pointed modal Abelian logic with respect to the variety of pointed modal Abelian $\ell$-groups.
\end{abstract}

\section{Introduction}\label{section:Introduction}

\emph{Abelian logic}, independently introduced by Meyer \& Slaney \cite{Meyer-Slaney:AbelianLogic} and Casari \cite{Casari:ComparativeLogics}, is a well-known contraclassical, paraconsistent logic. It is also referred to as the \emph{logic of Abelian $\ell$-groups} \cite{Butchart-Rogerson:Abel} or \emph{Abelian Group Logic} \cite{Paoli:LogicGroups} since the matrix models of Abelian logic are based on Abelian $\ell$-groups, their positive cones being the filters of designated elements.
A primary motivation for studying Abelian logic is its connection to \emph{Łukasiewicz logic}, since there exists a well-known categorical equivalence between the categories of strongly pointed Abelian $\ell$-groups and $\mathsf{MV}$-algebras, the algebras of Łukasiewicz logic \cite{Cignoli-Ottaviano-Mundici:AlgebraicFoundations}.       

Recently, \emph{modal Abelian logic} has also gained some attention, particularly from the perspective of its real-valued Kripke semantics \cite{Diaconescu-Metcalfe-Schnuriger:RealModalpre,Diaconescu-Metcalfe-Schnuriger:RealModal, Metcalfe-Tuyt:MonadicLogicAbelianGroups}. Following this, we investigate \emph{pointed real-valued modal Abelian logic} from an algebraic perspective. Our motivation for this focus is two-fold. First, it structurally parallels real-valued modal Łukasiewicz logic. Exploiting this connection, \emph{pointed} modal Abelian logic (albeit not explicitly named therein) is already used in \cite[Sect.\ 2.3]{Diaconescu-Metcalfe-Schnuriger:RealModal} to directly analyze real-valued modal Łukasiewicz logic \cite{HansoulTeheux:ModalLukasiewicz}. Specifically, their analysis requires the inclusion of a designated constant $c$ to act as a fixed point under the $\Box$-modality (i.e., $\Box c = c$), which serves as the falsum for Łukasiewicz logic. 
Furthermore, by substituting bounded $\mathsf{MV}$-algebras for Abelian $\ell$-groups, our framework lifts the traditional restriction to the bounded interval $[0,1]$, allowing formulas to instead be evaluated over $\mathbb{R}$. Second, this framework allows us to introduce a modal extension of \emph{Łukasiewicz unbound logic}. Originally defined in \cite{Cintula-Jankovec-Noguera:SuperabelianLogics}, this logic is strongly complete with respect to the Abelian $\ell$-group of real numbers equipped with the designated element $-1$. Consequently, by equipping our real-valued Kripke models with the same constant $-1$, we provide a natural modal semantics for Łukasiewicz unbound logic. 
Despite these connections, the \emph{algebraic} aspects of pointed modal Abelian logic have hitherto remained unexplored.

In this article we close this gap, introducing the variety of (negatively) \emph{pointed modal Abelian $\ell$-groups} $\pmAb$ (Definition~\ref{definition:ModalAbelienellGroup}) and explore its relationship to real-valued Kripke semantics of pointed modal Abelian logic (with constant $-1$). We do this by constructing complex algebras in $\pmAb$ from Kripke frames (entailing an algebraic soundness result, Proposition~\ref{proposition:SoundnessFullVariety}) and vice versa. Our first main result is a Truth Lemma for \emph{strongly} negatively pointed modal Abelian $\ell$-groups (Corollary~\ref{Lemma:TruthLemma}). Using this, we characterize the validities of pointed modal Abelian logic via an equational condition \eqref{InfinitaryRule} on the variety $\pmAb$, which may be seen as a form of (infinitary) \emph{algebraic completeness} (Theorem~\ref{t:completeness}). 

The primary technical challenges to overcome here stem from the algebraic behavior of strong units and the non-trivial intersection of maximal $\ell$-ideals. In general, a non-zero element of a negatively pointed modal Abelian $\ell$-group may belong to every maximal $\ell$-ideal. This effectively prevents us from using such an element as a falsifier for a formula, since no world in the corresponding canonical frame would evaluate it as non-zero. The rule \eqref{InfinitaryRule} ensure this does not happen. Furthermore, this rule allows us to argue via \emph{strongly} pointed algebras in our proof of Theorem~\ref{t:completeness}.
While it is a well-known fact that the property of possessing a strong unit is not first-order definable, enforcing a strong unit condition is justified within our framework since it corresponds directly to the restriction to \emph{bounded} valuations in real-valued Kripke models (required to guarantee the semantics of the $\Box$-modality is well-defined). We leave it up for future research to determine whether these issues could be circumvented in other ways. 

$\bullet$ \textbf{Related Work.} Prior research on \emph{real-valued modal logic} \cite{Diaconescu-Metcalfe-Schnuriger:RealModal, Diaconescu-Metcalfe-Schnuriger:RealModalpre} and \emph{modal \L ukasiewicz logic} \cite{HansoulTeheux:ModalLukasiewicz} has been a main inspiration for the current paper. More specifically, we build on the definition of (pointed) real-valued modal logic of the former and introduce an infinitary rule similar to the latter (here, let us also point out \cite{Baratella2018} about \emph{continuous modal logic} and \cite{LucasMio2022} about \emph{modal Riesz spaces}). More generally, our work stands in the research tradition on \emph{many-valued modal logic} (e.g.\ \cite{Fitting1991, Bou2009, HansoulTeheux:ModalLukasiewicz,Busaniche2022, Rodriguez2021}), note that we are using crisp accessibility relations (like \cite{Rodriguez2021} for Gödel and \cite{HansoulTeheux:ModalLukasiewicz} for \L ukasiewicz modal logic).  
Naturally the current paper is also related to the \emph{algebraic} theory of \emph{lattice-ordered Abelian groups} and \emph{$\mathsf{MV}$-algebras}  \cite{Cignoli-Ottaviano-Mundici:AlgebraicFoundations}. In particular, for more context about \emph{pointed} Abelian $\ell$-groups and their logic see \cite{Cintula-Jankovec-Noguera:SuperabelianLogics,Young:VarietiesPointedAbelianLGgroups,Jankovec:SubvarietiesPointedAbelianLGroups}.

$\bullet$ \textbf{Structure of the article.} In Section~\ref{section: AbelianLogic}, we provide some preliminaries about modal Abelian $\ell$-groups and proving some helpful results about $\ell$-ideals. In Section~\ref{Section:SemanticsAndModalAlgebras}, we introduce the real-valued relational semantics of pointed modal Abelian logic (Subsection~\ref{subsection:RelationalSemantics}) as well as the variety of pointed modal Abelian $\ell$-groups (Subsection~\ref{subsection:ModalAbelianellGroups}). We also explain how to relate these two via complex algebras and canonical frames (Subsection~\ref{subsection:AlgebrasAndFrames}). Section~\ref{section:TruthLemmaCompleteness} contains the main results of the article, namely a Truth Lemma for strongly pointed algebras (Subsection~\ref{subsection:TruthLemma}) and the infinitary algebraic completeness derived thereof (Subsection~\ref{subsection:TowardsCompleteness}). We discuss potential directions for future research in the concluding Section~\ref{section:Conclusion}.

\section{Pointed Abelian \texorpdfstring{$\ell$-groups}{l-groups}}
\label{section: AbelianLogic}

In this (mostly preliminary) section, we recall some important facts about pointed Abelian $\ell$-groups, focusing in particular on structures equipped with a negative strong unit (the reader should feel free to skip some of the more technical proofs here until we refer back to the corresponding results later on).
We give a brief exposition of $\ell$-ideal theory, and analyze the structure of $\ell$-ideals in the presence of such a designated unit. Adapting techniques from the theory of $\mathsf{MV}$-algebras, we characterize the radical, the intersection of all maximal $\ell$-ideals, via infinitesimal elements. Lastly, we establish a pivotal correspondence between the maximal $\ell$-ideals of a strongly negatively pointed Abelian $\ell$-group and its homomorphisms into the pointed Abelian $\ell$-group of reals $\alg R_{-1}$, which will be an important `bridge' between the abstract algebraic syntax and our concrete real-valued semantics.   

\emph{Abelian logic} (e.g.\ see \cite{Meyer-Slaney:AbelianLogic,Casari:ComparativeLogics,Cintula-Jankovec-Noguera:SuperabelianLogics}) can be elegantly axiomatized as the negation-free fragment of $\mathsf{FL}_e$ (i.e.\ Full Lambek calculus with exchange) extended by the axiom of relativity  $\big((\varphi \imp \psi) \imp \psi\big)  \imp \varphi$. Due to our algebraic focus, here we only treat Abelian logic via its algebraic models, i.e.\ Abelian $\ell$-groups. Nevertheless, let us mention that the logics corresponding to the varieties defined in this paper inherit similar axiomatizations, since they are superabelian (see \cite{Cintula-Jankovec-Noguera:SuperabelianLogics}). 

\begin{definition}[\sf{Abelian $\ell$-group}]\label{def:AbelianLGroup}
An algebra $\alg A=\tuple{A,+,-,\lor,\land,0}$ is an \emph{Abelian $\ell$-group} if $\tuple{A,+,-,0}$ is an Abelian group, $\tuple{A,\lor,\land}$ is a lattice and $\alg A$ satisfies the monotonicity condition $x \leq y \Rightarrow x+z \leq y+z$. 
\end{definition}

When brackets are omitted the lattice operations take precedence over the group operations, e.g.\ $x+y \land z$ is understood as $x+(y \land z)$. For $a \in A$ we define the \emph{absolute value} $|a| := a \lor -a$ and we let $n \cdot a:=a+\cdots+a$ denote the $n$-fold sum for every $n \in \mathbb{N}$. For $z \in \mathbb Z {\setminus} \mathbb N$ we set ${z \cdot a:=(-z) \cdot (-a)}$.

\begin{definition}[\sf{Strong Unit}]\label{def:StrongUnit}
Let $\alg A$ be an Abelian $\ell$-group. We say that $p \in A$ is a \emph{strong unit} of $\alg A$ if
\begin{equation} \label{e:StrongUnitGeneral} \tag{\textsf{SU}}
    \forall a \in A \colon \exists z \in \mathbb{Z} \colon z \cdot p \geq a. 
\end{equation}
Note that for a strictly negative strong unit $p < 0$ the condition \eqref{e:StrongUnitGeneral} simplifies to 
\begin{equation} \label{e:StrongUnitNegative} \tag{\textsf{SU}-}
    \forall a \in A \colon \exists n \in \mathbb{N} \colon n \cdot p \leq a. 
\end{equation}  
An algebra $\alg A_p=\tuple{A,+,-,\lor,\land,0,p}$ is a \emph{pointed Abelian $\ell$-group} if $\tuple{A,+,-,\lor,\land,0}$ is an Abelian $\ell$-group and $p \in A$. 
It is \emph{strongly pointed} if $p$ is a strong unit and \emph{negatively pointed} if $p \leq 0$. 
\end{definition}

\begin{remark}
    Traditionally, a strong unit is required to be strictly positive (e.g.\ see \cite[p.\ 300]{Birkhoff:LatticeTheory}). Note that we use $\mathbb{Z}$ in \eqref{e:StrongUnitGeneral}, which allows negative strong units. This distinction purely serves for notational convenience: since we primarily work in the negative cone, it avoids a proliferation of minus signs in our proofs. While our definition differs from the standard one, the structural role is identical: $p < 0$ is a strong unit in our sense if and only if its additive inverse $-p > 0$ is a standard strong unit.
\end{remark}

As a standard example, let $\alg R_{-1}=\tuple{\mathbb R,+,-,\lor,\land,0,-1}$ be the Abelian group of real numbers with addition together with its linear order and the designated point $-1$. This clearly is a strongly negatively pointed Abelian $\ell$-group (in fact, it is our \emph{algebra of truth-degrees} later on).

The following properties of Abelian $\ell$-groups are immediate consequences of the well-known fact that the variety of Abelian $\ell$-groups is generated by the $\ell$-group of integers $\mathbf{Z}$ (see \cite[Thm. 6.1]{Anderson-Feil:LatticeOrderedGroups}).

\begin{lemma} \label{l:identities}
Let $\alg A=\tuple{A,+,-,\lor,\land,0}$ be an Abelian $\ell$-group and $a,b \in A$.  
    \begin{enumerate}[label=(\roman*)]
        \item $a= a\land 0+a \lor 0$, 
        \item $\big((a-b) \lor 0\big) \land \big((b-a) \lor 0\big)=0$, \label{a:prime}  
        \item $(a \lor b)+c=(a+c) \lor (b+c)$,
        \item $(a-b) \lor 0=a-a\land b$,
        \item $|a+b|\leq |a|+|b|$.
    \end{enumerate}
\end{lemma}

\begin{definition}[\textsf{$\ell$-ideal}]
An \emph{$\ell$-ideal} of an Abelian $\ell$-group $\alg A$ is an $\ell$-subgroup $I \subseteq \alg A$ that is \emph{convex} (equivalently $0 \leq |x| \leq |i|$ for some $i \in I$ implies $x \in I$).
An $\ell$-ideal $I$ is \emph{prime} if it satisfies 
\begin{equation} \label{e:prime}\tag{\textsf{PI}}
    a \land b \in I \quad \Rightarrow \quad a \in I \text{ or } b \in I.
\end{equation}
\end{definition}

First, let us note that for an $\ell$-ideal $I$ and $a \in A$ by definition it holds that $a \in I$ iff $|a| \in I$ iff $-|a| \in I$.
One can easily prove (using Zorn's Lemma) that for each element $a \in A{\setminus}\{0\}$ there exists an $\ell$-ideal which is maximal among all the $\ell$-ideals which do not contain $a$. Every such $\ell$-ideal is prime and we refer to them as \emph{values} of $a$ \cite[Thm. 2.4.2]{Steinberg:LatticeModules}.

It is well known that \emph{maximal} $\ell$-ideals are prime (again see \cite[Thm. 2.4.2]{Steinberg:LatticeModules}), although in general an Abelian $\ell$-group $\alg{A}$ need not have any maximal ideals. However, if $p \in A$ is a strong unit, the maximal $\ell$-ideals correspond to values of $p$ (since any $\ell$-ideal containing $p$ is improper). In fact, in the presence of a strong unit every proper $\ell$-ideal $I \subseteq \alg{A}$ can be extended to a maximal $\ell$-ideal $M \supseteq I$ (which is again easily shown via Zorn's Lemma).

Let $\alg A$ be an Abelian $\ell$-group and $S \subseteq A$.
By $C(S)$ we denote the $\ell$-ideal generated by $S$.
It is well-known \cite[Thm. 2.2.4 (b)]{Steinberg:LatticeModules} that
\begin{equation} \label{e:generated ideal}\tag{\textsf{IG}}
    C(S):=\{x \in A \mid |x|\leq |b_1|+\dots+|b_n| \text{ for some } b_1,\dots, b_n \in S \}.
\end{equation}

Let us also recall from \cite[Thm. 2.2.4 (d)]{Steinberg:LatticeModules}, that for two sets $S,T \subseteq A$ we have
\begin{equation*} 
    C(S) \lor C(T)=C(S \cup T)=\{x \in A \mid |x|\leq |b_1|+\dots+|b_n| \text{ for some } b_1,\dots, b_n \in S \cup T \}.
\end{equation*}
Therefore, for a principal $\ell$-ideal $J = \langle a\rangle= \{x \mid \exists n\in \mathbb N \colon |x|\leq n \cdot |a|\}$, we obtain 
$$I \lor J=I \lor \{x \mid \exists n \in \mathbb N \colon |x|\leq n\cdot|a|\}=\{x \mid \exists n \in \mathbb N \; \exists i \in I \colon  |x|\leq i+n\cdot|a|\}.$$
In particular, this implies that for every maximal $\ell$-ideal $I$ of an Abelian $\ell$-group $\alg A$ and $a \notin I$:
\begin{equation}\label{e:ideal-join}\tag{\textsf{IG}$\ast$}
A = \{x \mid \exists n \in \mathbb N \colon \exists i \in I \colon  |x|\leq i+n \cdot |a|\}.
\end{equation}
For a strongly pointed Abelian $\ell$-group $\alg A_p$ we define the \emph{radical} of $\alg A_p$ as the intersection of all maximal ideals of $\alg A_p$ and denote it by $\rad(\alg A_p)$. Equivalently, $\rad(\alg A_p)$ is the intersection of all values of $p$. We call $a \in A$ is \emph{infinitely small} (or \emph{infinitesimal}) if 
\begin{equation} \tag{\textsf{INF}} \label{e:infinitesimalGen}
    \forall n \in \mathbb{N}\colon n \cdot |a|\leq |p| .
\end{equation}
In the case where $a,p \leq 0$, Equation~\eqref{e:infinitesimalGen} simplifies to 
\begin{equation} \tag{\textsf{INF}-} \label{e:infinitesimalNegative}
\forall n \in \mathbb{N}\colon n \cdot a-p \geq 0.
\end{equation}
One can easily check that this is also equivalent to 
$\forall k \in \mathbb{N}\colon 2^k \cdot a-p \geq 0$.

We now show that the radical of a strongly pointed Abelian $\ell$-group consists exactly of its infinitesimal elements. The proof of the following is similar to the one of \cite[Prop. 3.6.4]{Cignoli-Ottaviano-Mundici:AlgebraicFoundations}.

\begin{lemma}\label{l:infinitesimal}
    Let $\alg A_p$ be a strongly pointed Abelian $\ell$-group. Then $a \in \rad(\alg A_p)$ iff $a$ is infinitely small.
\end{lemma}
\begin{proof}
    If $a \notin \rad(\alg A_p)$ there is a maximal $\ell$-ideal $I$ such that $a \notin I$.
    Consequently, $p$ is in the $\ell$-ideal generated by $\{a\}\cup I$ and thus by \eqref{e:ideal-join} there are $n \in \mathbb N$ and $i \in I$ such that $0<|p|\leq n \cdot |a|+i$, 
     which is equivalent to $|p|-n \cdot |a| \leq i$. Assume towards contradiction that $a$ is infinitely small. Then $|p|-(n+1)\cdot |a| \geq 0$ for each $n \in \mathbb N$, equivalently $|p|-n\cdot |a| \geq |a|$ for each $n \in \mathbb N$. Thus altogether we have $0 \leq |a| \leq |p|-n\cdot |a| \leq i$, which by convexity of the $\ell$- ideal $I$ implies $a \in I$, a contradiction. 

    Conversely, assume $a \in A$ is not infinitely small. Then there is $n \in \mathbb N$ such that
    $n \cdot |a|-|p| \nleq 0$, which implies $(n \cdot |a|-|p|) \lor 0 > 0$. Let $P$ be a value of $(n \cdot |a|-|p|) \lor 0$.
    By Lemma~\ref{l:identities} \ref{a:prime} we have
    $$
    \big((n \cdot |a|-|p|) \lor 0\big) \land \big((|p|-n \cdot |a|) \lor 0\big)=0.
    $$
    Since $P$ is prime, it follows that $(|p|-n \cdot |a|) \lor 0 \in P$. Let $M$ be a maximal $\ell$-ideal containing $P$ and towards contradiction assume that $a \in M$. Then $(n \cdot |a|) \lor 0 \in M$ and $$(n \cdot |a|) \lor 0\geq (n \cdot |a|-|p|) \lor 0  > 0.$$ Therefore, $(n \cdot |a|-|p|) \lor 0 \in M $. 
    Consequently, by Lemma~\ref{l:identities}(i) we find $$ (|p|-n \cdot |a|) \lor 0-(n \cdot |a|-|p|) \lor 0=(|p|-n \cdot |a|) \lor 0+(|p|-n \cdot |a|) \land 0 =(|p|-n \cdot |a|) \in M.$$
    Since $p \notin M$ we have $|p| \notin M$ and thus also $n \cdot |a| \notin M$. However, this implies $|a| \notin M$, and thus lastly $a \notin M$, a contradiction our assumption $a \in M$ which shows $a \notin \rad(\alg{A}_p)$.   
\end{proof}

\begin{lemma} \label{l:max ideal correspondence}
    Let $\alg A_p$ be a strongly negatively pointed Abelian $\ell$-group and $I$ an $\ell$-ideal of $\alg A$. Then $I$ is maximal if and only if $I=v^{-1}(0)$ for some homomorphism $v \in \Hom(\alg A_p,\R_{-1})$. 
\end{lemma}

\begin{proof}
$(\Rightarrow)$:
Let $I$ be maximal. By Hölder's Theorem \cite[Thm. 2.3.10]{Steinberg:LatticeModules} $\alg A_p/I$ embeds into $\R_r$ for some $r \in \mathbb R$. This gives us $v \in \Hom(\alg A_p,\R_{r})$. 
Since $p$ is a strong unit of $\alg A$, $p \notin I$ and $v(p) \neq 0$. Since $p$ is negative, also $v(p)\leq 0$.
Therefore $v(p)=r<0$. Lastly, since $\R_{r_1} \cong \R_{r_2}$ for any two negative $r_1,r_2 \in \mathbb R$, we can assume without loss of generality that $v \in \Hom(\alg A_p,\R_{-1})$. 

$(\Leftarrow)$: For every $v \in \Hom(\alg A_p,\R_{-1})$, the subset $v^{-1}(0)$ is an $\ell$-ideal of $\alg{A}_p$ and any $\ell$-subgroup of $\R_{-1}$ is simple, which ensures that the ideal $v^{-1}(0)$ has to be maximal.
\end{proof}

This correspondence between infinitely small elements, maximal $\ell$-ideals and homomorphisms into $\alg{R}_{-1}$ will be important later on, since $\alg{R}_{-1}$ is the algebra of truth-degrees for pointed modal Abelian logic introduced in the next section (in particular, \emph{canonical frames} are structures based on the collections of homomorphisms $\Hom(\alg{A}_p, \alg{R}_{-1})$).     

\section{Pointed Modal Abelian Logic: Frames and Algebras}
\label{Section:SemanticsAndModalAlgebras}
We first introduce relational (Kripke) semantics for pointed modal Abelian logic similar to the real-valued logic of \cite{Diaconescu-Metcalfe-Schnuriger:RealModal} and other many-valued modal logics on crisp frames (Subsection~\ref{subsection:RelationalSemantics}). Next, we introduce the variety of pointed modal Abelian $\ell$-groups and establish some basic identities therein (Subsection~\ref{subsection:ModalAbelianellGroups}). Lastly, we explain how to `shift' between these two frameworks constructing complex algebras (in particular showing algebraic soundness) and canonical frames (Subsection~\ref{subsection:AlgebrasAndFrames}). This sets up the scene for our main results which we prove in the subsequent Section~\ref{section:TruthLemmaCompleteness}. 

\subsection{Relational Semantics}
\label{subsection:RelationalSemantics}
The \emph{language} $\Lang_p$ of pointed modal Abelian logic is defined inductively from a countable set $\Prop = \{ x_1, x_2, x_3, \dots \}$ of \emph{propositional variables}, the signature $\tuple{+, -, \vee, \wedge, 0,p}$ of pointed Abelian $\ell$-groups and the unary modality $\Box$.  
Let $\Fm{\Lang_p}$ denote the set of \emph{formulas} over the language $\Lang_p$. 

\begin{definition}[\sf Model]\label{definition:Model}
A \emph{model} $\mathfrak{M} = \tuple{W, R, \Val}$ consists of a \emph{Kripke frame} $\tuple{W,R}$ (i.e.\ a non-empty set $W$ with a binary relation $R \subseteq W^2$), together with a \emph{bounded propositional valuation}, by which we mean a map
$
\Val\colon W \times \Prop \to \mathbb{R}
$ 
such that for every propositional variable $x_i$ there exists a positive real number $r_i$ with $\{ \Val(u, x_i) \mid u \in W \} \subseteq [-r_i,r_i]$. 
\end{definition}

Note that this slightly differs from \cite{Diaconescu-Metcalfe-Schnuriger:RealModal}, where a propositional valuation is itself bounded. However, it is clear that our condition suffices to ensure that the following interpretation of $\Box$ is (still) well-defined.

\begin{definition}[\sf Satisfaction \& Validity]
Let $\mathfrak{M} = \tuple{W, R, \Val}$ be a model and $u \in W$. We inductively extend the valuation to $\Val \colon W \times \Fm{\Lang_p} \to \mathbb{R}$ by interpreting the non-modal connectives in the algebra $\alg{R}_{-1}$ in the standard way. Modal formulas are interpreted as follows:
$$
\Val(u, \Box\varphi) = \bigwedge_{uRv} \Val(v, \varphi). 
$$
(In particular, here we set $\bigwedge\varnothing := 0$.)
A formula $\varphi$ is \emph{satisfied} at $u \in W$ if and only if 
$$
\Val(u, \varphi) \geq 0
$$
and the formula is \emph{valid}, denoted $\vDash \varphi$, if this holds at every state of every model. 
\end{definition}

The following properties correspond to our axiomatization of pointed modal Abelian $\ell$-groups later on (Definition~\ref{definition:ModalAbelienellGroup}) and essentially yield (algebraic) \emph{soundness} later on.   

\begin{lemma}\label{lemma:Soundness}
For $\mathfrak{M} = \tuple{W, R, \Val}$ a model, $u \in W$, and $\varphi,\psi \in \Fm{\Lang_p}$, the following hold:
\begin{multicols}{2}
    \begin{enumerate}[label=(\roman*)]
        \item $\Val\big(u,\Box (\psi-\varphi)\big) \leq \Val(u,\Box \psi)-\Val(u,\Box \varphi)$,
        \item $\Val(u,\Box \varphi)+\Val(u,\Box \varphi)=\Val\big(u,\Box (\varphi+\varphi)\big)$,
        \item $\Val(u,\Box \varphi) \lor 0=\Val\big(u,\Box (\varphi\lor 0)\big)$,
        \item $\Val(u,\Box \varphi) \land \Val(u,\Box \psi)=\Val\big(u,\Box(\varphi \land \psi)\big)$,
    \item $\Val(u, |\Box p-p| \land |\Box \varphi| )=0$,
    \item $\Val(u, -\Box p)=\Val(u, \Box {-p})$.
    \end{enumerate}
\end{multicols}   
\end{lemma}
\begin{proof}
\emph{(i)}: By substituting $\chi=\psi-\varphi$ and adding $\Val(u,\Box \varphi)$ to both sides of the inequation, it suffices to prove 
    $\Val(u,\Box \chi)+ \Val(u,\Box \varphi)\leq \Val(u,\Box (\chi+\varphi))$, or equivalently
    $$\bigwedge_{uRv} \Val(v,\chi)+ \bigwedge_{uRv} \Val(v, \varphi) \leq \bigwedge_{uRv} \Val(v,\chi+\varphi).$$
This follows, since we have $\bigwedge_{uRv} \Val(v,\chi)+ \bigwedge_{uRv} \Val(v, \varphi) \leq  \Val(v,\chi+\varphi)$ for each $uRv$.
\\
\emph{(ii)}: Directly  follows from the identity $\bigwedge_{i\in I} r_i + \bigwedge_{i\in I} r_i = \bigwedge_{i\in I} r_i + r_i$ of $\alg{R}_{-1}$. 
\\    
\emph{(iii)}: Directly follows from the infinite distributivity law of $\ell$-groups, see \cite[Eq. 1.1.11]{Anderson-Feil:LatticeOrderedGroups}.
\\
\emph{(iv)}: Directly follows from the law $\bigwedge_{i \in I}r_i \wedge \bigwedge_{i \in I} s_i = \bigwedge_{i\in I} r_i \wedge s_i$ of $\alg{R}_{-1}$.
\\
\emph{(v) \& (vi)}: We distinguish cases (a) $\exists v\colon uRv$ and (b) $\not\exists v \colon uRv$. 
In case (a) note $\Val(u,\Box p)=-1=\Val(u,p)$ and thus $\Val(u, |\Box p-p| \land |\Box \varphi| )=0$ and $\Val(u,-\Box p)=1=\Val(u,\Box - p)$.
In case (b) we have $\Val(u,\Box \varphi)= \bigwedge \varnothing = 0$, thus $\Val(u, |\Box p-p| \land |\Box \varphi| )=0$. Furthermore, $\Val(u, -\Box p)=  0 = \Val(u, \Box {-p})$.
\end{proof}

\subsection{Modal Abelian \texorpdfstring{$\ell$-groups}{l-groups}}
\label{subsection:ModalAbelianellGroups}

We aim to relate the relational semantics introduced in the previous subsection to the variety of pointed Abelian $\ell$-groups with unary operators defined as follows.

\begin{definition}[\textsf{Modal Abelian $\ell$-group}]\label{definition:ModalAbelienellGroup}
An algebra $\tuple{\alg A,\Box}$ is a \emph{modal Abelian $\ell$-group} if $\alg A$ is an Abelian $\ell$-group and $\Box\colon A \to A$ is a unary operation which satisfies for each $x,y \in A$ the following:
\begin{multicols}{2}
\begin{enumerate}[label=(\textsf{A\arabic*})]
    \item\label{A1} $\Box (y-x) \leq \Box y - \Box x,$
    \item\label{A2} $\Box(x + x) = \Box x + \Box x,$ 
    \item\label{A3} $\Box(x \lor 0)=\Box x \lor 0,$ \label{a:Join0}
    \item\label{A4} $\Box (x \land 0) = \Box x \land 0.$ \label{a:Meet}
\end{enumerate}
\end{multicols}
If $\alg A_p$ is a pointed Abelian $\ell$-group, we say that $\tuple{\alg A_p, \Box}$ is a \emph{pointed modal Abelian $\ell$-group} if it furthermore satisfies the following: 
\begin{multicols}{2}
\begin{enumerate}[label=(\textsf{A\arabic*})]\setcounter{enumi}{4}
    \item\label{A5} $|\Box x| \land |\Box p-p|=0$, \label{a:p}
    \item\label{A6} $\Box {-p}=-\Box p$. \label{a:-p}
\end{enumerate}
\end{multicols}
We denote the varieties of modal Abelian $\ell$-groups and of negatively pointed modal Abelian $\ell$-groups by $\mAb$ and $\pmAb$, respectively.
A pointed modal Abelian $\ell$-group $\tuple{\alg A_p, \Box}$ is \emph{strongly negatively pointed} if $\alg A_p$ is strongly negatively pointed.
\end{definition}

\begin{example}
    Let $\alg A_p$ be a pointed Abelian $\ell$-group. There are always at least two `trivial' ways to expand this structure to a modal pointed Abelian $\ell$-group, namely the identity $\Box_1 x := x$ and the constant zero $\Box_2 x := 0$. It is straightforward to verify that both $\tuple{\alg A_p, \Box_1}$ and $\tuple{\alg A_p, \Box_2}$ satisfy axioms \ref{A1}--\ref{A6}. 
\end{example}

Clearly $\pmAb$ is a conservative expansion of $\mAb$, as in every $\alg{A} \in \mAb$ one can interpret the additional constant symbol as $p^\alg{A}\coloneq 0^\alg{A}$. Then Axioms \ref{a:p}--\ref{a:-p} are trivially satisfied by (i) in the following.

\begin{lemma} \label{l:modal basics}
        The following hold in all modal Abelian $\ell$-groups.
    \begin{multicols}{2}
    \begin{enumerate}[label=(\roman*)]
        \item  $\Box0=0$, \label{a:0}
        \item  $\Box (2^k \cdot x) = 2^k \cdot \Box x$ for each $k \in \mathbb N$, \label{a:density}
        \item  $x \geq 0 \Rightarrow\Box x \geq 0$, 
        \item  $x \leq y \Rightarrow\Box x \leq \Box y$. 
    \end{enumerate}
    \end{multicols}
Moreover, the following hold in all pointed modal Abelian $\ell$-groups.
\begin{multicols}{2}
\begin{enumerate}[label=(\roman*)]\setcounter{enumi}{4}
    \item $\Box(x+p)=\Box x+ \Box p$ and \\ $\Box(x-p)=\Box x-\Box p,$
    \item $\Box(x+z \cdot p)=\Box x+z \cdot \Box p$ for each $z \in \mathbb Z$, \label{a:adding p}
    \item $\Box(x \lor p)=\Box x \lor \Box p$,
    \item $p \leq \Box p$, 
    \item $\Box (-k \cdot p)=-\Box (k \cdot p)$ for each $k \in \mathbb N$.
\end{enumerate}
\end{multicols}    
\end{lemma}

\begin{proof}
 \emph{(i)}: Follows from \ref{A2} with $x=0$. 
 \\
\emph{(ii)}: Follows from iterating \ref{A2}.
\\
\emph{(iii)}: Follows from \ref{A3}.
\\
\emph{(iv)}: If $x \leq y$, we have $0 \leq y-x$ and thus $0 \leq \Box(y-x)$ by the previous item. Therefore $0 \leq \Box(y-x) \leq \Box y - \Box x$ by \ref{A1}, which implies $\Box x \leq \Box y$.
\\
 \emph{(v)}:  Using \ref{A1} and \ref{a:-p} we compute 
    $$
    \Box(x+p) \leq \Box x-\Box(-p)=\Box x+\Box p=\Box(x+p-p)+\Box p \leq \Box(x+p) -\Box p +\Box p=\Box(x+p).
    $$
    Therefore, $\Box(x+p)=\Box x+\Box p$ for each $x \in A$.
    Consequently by substitution, $\Box x=\Box(x-p+p)=\Box (x-p)+\Box p$ for each $x \in A$.
\\
\emph{(vi)}: Follows from iterating the previous item.
\\
 \emph{(vii)}: Using \ref{A3}, (vi) and Lemma~\ref{l:identities} (iii) we compute 
    \begin{align*}
        \Box(x \lor p)&=\Box\big(((x-p) \lor 0)+p\big) & \text{(Lemma~\ref{l:identities}(iii))}\\
        &=\Box\big((x-p) \lor 0\big)+\Box p & \text{(Item (v))}\\
        &=\big(\Box(x-p) \lor 0\big) +\Box p & \text{(\ref{A3} \& Item (v))}\\
        &=\big(\Box(x-p)+\Box p\big) \lor \Box p=\Box x \lor \Box p.& \text{(Lemma~\ref{l:identities}(iii) \& Item (v))}
    \end{align*}
\emph{(viii)}: Due to Lemma~\ref{l:identities} (iv) \& (v) we compute that $(|\Box p|-|\Box p-p|) \leq |p|$ and furthermore that $|\Box p|-|\Box p| \land |\Box p-p|=(|\Box p|-|\Box p-p|) \lor 0$.
    Therefore by \ref{A5} we obtain
    $$|\Box p|=|\Box p|-|\Box p| \land |\Box p-p|=(|\Box p|-|\Box p-p|) \lor 0\leq |p|.$$
    Consequently, by applying \ref{A6},(vii) and $p \leq 0$ compute $$|\Box p|=-\Box p \lor \Box p=\Box (-p) \lor \Box p=\Box(-p \lor p)=\Box |p|=\Box (-p)=-\Box p$$
    which shows $p \leq \Box p$ as desired.
\\    
\emph{(ix)}: Is a straightforward consequence of (i) and (vi) with $x = 0$. 
\end{proof}

\begin{remark}
    As already discussed in Lemma~\ref{lemma:Soundness}, by applying subtraction and substitution, it is easy to see that axiom \ref{A1} can be equivalently replaced with
    \begin{equation} \tag{\textsf{A1}$\ast$} \label{e:superadditivity}
        \Box x + \Box y \leq \Box(x+y).
    \end{equation}

Let us also note that we could replace \ref{A5} and \ref{A6} with a single axiom
    \begin{equation} \tag{\textsf{A5}$\ast$}
        |\Box x| \land |\Box(x+x-p)-\Box x+\Box x -p| =0
    \end{equation}
Although this axiom is harder to verify, together with \ref{A2} it provides an equational counterpart to the modal $\mathsf{MV}$-algebra axioms $\Box(x \oplus x)=\Box x \oplus \Box x$ and $\Box(x \odot x)=\Box x \odot \Box x$ used in \cite{HansoulTeheux:ModalLukasiewicz}.
\end{remark}

\subsection{Complex Algebras and Canonical Frames}
\label{subsection:AlgebrasAndFrames}

As usual, to `shift' between the relational and algebraic frameworks introduced in the previous two subsections, we define complex algebras and canonical frames, starting with the former.

\begin{definition}[\sf Complex Algebra]\label{definition:ComplexAlgebra}
Let $\tuple{W,R}$ be a Kripke frame. The corresponding \emph{complex algebra} is given by $\mathfrak{C}_{(W,R)} = \langle b\alg{R}_{-1}^W, \Box_R \rangle$ where $b\alg{R}_{-1}^W \leq \alg{R}_{-1}^W$ denotes the subalgebra of the product $\alg{R}_{-1}^W$ defined by the subuniverse of \emph{bounded maps} $\{ f \colon W \to \mathbb{R} \text{ bounded}\}$ and the operator $\Box_R$ is defined by 
$$
(\Box_Rf)(u) = \bigwedge_{uRv} f(v).
$$
(Note that $\Box f$ is a well-defined bounded map if $f$ is bounded.)
\end{definition}

It is clear from Lemma~\ref{lemma:Soundness} that all complex algebras are pointed modal Abelian $\ell$-groups. Furthermore, the map of constant value $-1$ is a strong unit. From this we get the following. 

\begin{proposition}[\sf Algebraic Soundness]\label{proposition:SoundnessFullVariety}
If $\pmAb \vDash \varphi \geq 0$ then $\vDash \varphi$. 
\end{proposition}
\begin{proof}
If $\not\vDash\varphi$ then there is a model $\tuple{W,R,\Val}$ such that $\Val(u, \varphi) < 0$. Then $e_\Val \colon \Prop \to \mathfrak{C}_{\tuple{W,R}}$ defined by $e_\Val(x_i) = \Val(-,x_i)$ is an assignment under which $\varphi \geq 0$ fails in the complex algebra $\mathfrak{C}_{\tuple{W,R}}$.
\end{proof}

Similarly, we move from $\pmAb$-algebras to Kripke frames via the following construction based on the set of all homomorphisms into our algebra of truth-values $\alg{R}_{-1}$.   

\begin{definition}[\sf Canonical Frame]\label{definition:CanonicalFrames}
    Let $\tuple{\alg{A}_p,\Box}$ be a strongly negatively pointed modal Abelian $\ell$-group and let $u,v \in \Hom(\alg{A}_p,\R_{-1})$ be homomorphisms. We define 
    $$
    uR_\Box v \quad \text{if and only if} \quad u(\Box a) \geq 0 \Rightarrow v(a) \geq 0 \text{ for all } a \in \alg{A}_p.
    $$
We call $\mathfrak F_{\tuple{\alg{A}_p,\Box}}:=\tuple{\Hom(\alg{A}_p,\R_{-1}),R_\Box}$ the \emph{canonical frame} of $\tuple{\alg{A}_p,\Box}$.
\end{definition}

Note that from any \emph{algebraic evaluation} $e \colon \Prop \to \alg{A}_p$ we can define a propositional valuation $\Val_e \colon \Hom(\alg{A}_p, \alg{R}_{-1}) \times \Prop \to \mathbb{R}$ on the canonical frame $\mathfrak{F}_{(\alg{A}_p, \Box)}$ via $\Val_e(u,x_i) = u(e(x_i))$. However, these are \emph{not} necessarily bounded valuations, that is, they do not necessarily define models according to Definition~\ref{definition:Model}. However, note that these valuations are bounded if $p$ is a \emph{negative strong unit}, since then $e(x_i) \in \alg{A}_p$ is bounded by $n\cdot p \leq e(x_i) \leq -n \cdot p$ for some $n \in \mathbb{N}$. Therefore, every homomorphism $u \colon \alg{A}_p \to \alg{R}_{-1}$ satisfies $u(e(x_i)) \in [-n, n]$. In the next section, we prove that the corresponding Truth Lemma also holds for strongly negatively pointed modal Abelian $\ell$-groups. 

\section{Pointed Modal Abelian Logic, Algebraically}
\label{section:TruthLemmaCompleteness}
In Subsection~\ref{subsection:TruthLemma}, we establish the Truth Lemma with respect to the constructions of the previous subsection for \emph{strongly} negatively pointed modal Abelian $\ell$-groups (Corollary~\ref{Lemma:TruthLemma}). For this, we introduce some terms $t_d$ which can `separate' elements of $\alg{R}_{-1}$. 

In Subsection~\ref{subsection:TowardsCompleteness}, we introduce an infinitary derivability relation with respect to $\pmAb$ and prove completeness of real-valued modal logic with respect to this relation (Theorem~\ref{t:completeness}).  

\subsection{The Truth Lemma}
\label{subsection:TruthLemma}
First let us note that the following is easily derived from \ref{A5} and the fact that $\R_{-1}$ is totally ordered.
\begin{lemma} \label{l:dichotomy u}
    Let $\tuple{\alg{A}_p, \Box}$ be a pointed modal Abelian $\ell$-group and let $u \in \Hom(\alg{A}_p,\R_{-1})$ be a homomorphism. Then $u(\Box a)=0$ for each $a \in \alg{A}_p$ or $u(\Box p)=u(p) = -1$.
\end{lemma}

The set $ D:=\{\frac{z}{2^k} \mid z \in \mathbb Z, k \in \mathbb N\}$ is a dense subset of $\mathbb R$. For each $d \in D$ we fix some representation $d = \frac{z}{2^k}$ and define $t_d$ to be a unary term  
$$t_d(x)=(2^k \cdot x + z \cdot p) \land 0.$$ 
The following summarizes the main properties of these terms for our purposes.
\begin{lemma} \label{l:term separator}
The terms $t_d$ defined above satisfy the following.  
    \begin{enumerate}[label=(\roman*)]
        \item  $t_d^{\alg{A}_p}(x) \leq 0$ for every $x \in {\alg{A}_p}$. 
        \item $t_d^{\R_{-1}}(r) = 0$ iff $r \geq d$.

        \item Assume $u(\Box y) \neq 0$ for some $y \in \alg{A}_p$. Then
        $u \big(t_d(\Box x)\big)=u \big(\Box t_d(x)\big)$.
        
    \end{enumerate}
\end{lemma}

\begin{proof}
    Let $d=\frac{z}{2^k}$ be the representation chosen in the definition of $t_d$. By definition, \emph{(i)} is obvious. 
    \emph{(ii)}: Observe that for each $r \in \mathbb R$ we have $t_d(r) \geq 0$ iff $2^k \cdot r + z \cdot (-1) \geq 0$, which is equivalent to $r-d\geq 0$.
\\
    \emph{(iii)}: First observe that by Lemma \ref{l:dichotomy u} and the assumption we obtain $u(\Box p)=u(p)$. Compute
    \begin{align*}
        u \big(t_d(\Box x)\big)&=u\big((2^k \cdot \Box x + z \cdot p) \land 0\big) &\text{(Def. $t_d$)} \\
        &=\big(u(2^k \cdot \Box x)+z \cdot u(p)\big) \land 0 &\text{($u$ homomorphism)}\\ 
        &=\big(u(2^k \cdot \Box x)+z \cdot u(\Box p)\big) \land 0 &\text{($u(\Box p) = p$)}\\
        &=u\big((2^k \cdot \Box x + z \cdot \Box p) \land 0\big)&\text{($u$ homomorphism)}\\
        &=u\big((\Box (2^k \cdot x) + z \cdot \Box p) \land 0\big)&\text{(Lemma~\ref{l:modal basics}(ii)) }\\
        &=u\big(\Box ( 2^k \cdot x + z \cdot p) \land 0\big) &\text{(Lemma~\ref{l:modal basics}(vi))}\\ 
        &=u\big(\Box (( 2^k \cdot x + z \cdot p) \land 0)\big)= u\big(\Box t_d (x)\big) &\text{(Axioms \ref{A1},\ref{A4} \& Def. $t_d$)}
    \end{align*}    
    as desired, finishing the proof.
\end{proof}

Among others, we use these terms to prove the following (recall the Definition~\ref{definition:CanonicalFrames} of the canonical frame $\tuple{\Hom(\alg{A}_p, \R_{-1}), R_\Box}$).  

\begin{lemma} \label{l:uRv}
    Let $\tuple{\alg{A}_p, \Box} \in \pmAb$ be strongly negatively pointed and $u,v \in \Hom(\alg{A}_p,\R_{-1})$. The following conditions are equivalent:

    \begin{enumerate} [label=(\arabic*)]
        \item $uR_\Box v$,
        \item  $u(\Box a)=0 \Rightarrow v(a) = 0$ for all $a \leq0$,
        \item $\Box^{-1} u^{-1}(0) \cap \{a \in A \colon a \leq 0\} \subseteq v^{-1}(0)$,
        \item $u(\Box a) \leq v(a)$ for all $a \in A$.
 
    \end{enumerate}
\end{lemma}

\begin{proof}
$(1)\Rightarrow(2)$: Let $a \leq 0$, and $u(\Box a) = 0$. The former yields $v(a) \leq 0$ and the latter implies together with $u R_\Box v$ that $v(a) \geq 0$.
\\
$(2) \Leftrightarrow (3)$ is obvious.
\\
    $(2) \Rightarrow (4)$: We proceed by contraposition.
    Assume there is $a$ such that $v(a)<u(\Box a)$. Then there is $d \in D$ such that $v(a)<d\leq u(\Box a)$. By Lemma \ref{l:term separator}(ii) we have $t_d^{\R_{-1}} ( v(a))=v(t_d(a)) < 0$ and 
    $t_d^{\R_{-1}}( u (\Box a))=0$. By Lemma \ref{l:term separator}(iii) this implies   $u (\Box t_d(a)) = 0$. Also, by Lemma \ref{l:term separator}(i) we have $t_d(a) \leq 0$.
\\
    $(4)\Rightarrow (1)$: If $u (\Box a) \leq v(a)$ for each $a \in A$ then clearly whenever $0 \leq u(\Box a)$ then also $0 \leq v(a)$.
\end{proof}

We show one more technical result before proving the main theorem of this subsection.
\begin{lemma} \label{l:almost ideal}
 Let $\tuple{\alg{A}_p, \Box}$ be a strongly negatively pointed modal Abelian $\ell$-group, $u \in \Hom(\alg{A}_p,\alg R_{-1})$, and let $I$ be the $\ell$-ideal generated by the set $S := \{x \in A \mid x \leq 0\} \cap \Box^{-1} u^{-1}(0)$.
    Let $a \in I$ such that $a \leq 0$. Then $a \in \Box^{-1} u^{-1}(0)$.
\end{lemma}

\begin{proof}
    Since $a$ is in the $\ell$-ideal generated by the set $S$, due to \eqref{e:generated ideal}, there exist $b_1,\dots, b_n \in S$ such that $\sum_{i=1}^{n} b_i \leq a \leq 0$.  
    Since by Lemma \ref{l:modal basics} (iv) the operator $\Box$ is order preserving, by iterating Equation \eqref{e:superadditivity} we have the following in $\alg A_p$:  
    $$\Box b_1+\dots+\Box b_n \leq \Box (b_1 +\dots +b_n) \leq \Box a \leq \Box 0=0. $$
    
    By applying $u$, we obtain:
     $$0=u(\Box b_1+\dots+\Box b_n) \leq u\big(\Box (b_1 +\dots +b_n)\big) \leq u(\Box a) \leq u(\Box 0)=0. $$
    Thus, $u(\Box a)=0$, which completes the proof.
\end{proof}

The following is crucial, essentially being the clause of modal formulas in the Truth Lemma. 

\begin{theorem}\label{theorem:TruthLemmaBox}
    Let $\tuple{\alg{A}_p, \Box}$ be a strongly negatively pointed modal Abelian $\ell$-group.
    Then $u(\Box a)= \bigwedge_{uR_\Box v} v(a)$ holds
     for every homomorphism $u\in \Hom(\A_p,\alg R_{-1})$. 
\end{theorem}
\begin{proof}
     We already know by Lemma~\ref{l:uRv}(4) that $u(\Box a) \leq  \bigwedge_{uR_\Box v} v(a)$ (from this we also see that the meet exists). 
     First, we consider the case $u(\Box a)=0$ for every $a \in \alg{A}_p$. Since for any $v \in \Hom(\A_p,\R_{-1})$ we have $v(p)=-1$, it follows there is no $v$ such that $uR_\Box v$. Consequently, $\bigwedge_{uR_\Box v} v(a)=0=u(\Box a)$.
     
     For the remaining case, by Lemma~\ref{l:dichotomy u} we know that $u (\Box p)=-1$.
     Towards contradiction, assume $u(\Box a) <  \bigwedge_{uR_\Box v} v(a)$. Since $D$ is dense there exists $d \in D$ such that $u(\Box a) < d \leq  \bigwedge_{uR_\Box v} v(a)$. Therefore, by Lemma \ref{l:term separator}(ii) \& (iii) we have $ t_d^{\R_{-1}} ( u(\Box a))=u(t_d (\Box a))=u(\Box t_d(a)) < 0$ and $t_d^{\R_{-1}} (v(a))=v( t_d(a)) = 0$ for each $uR_\Box v$. 
     By Lemmata \ref{l:uRv} \& \ref{l:max ideal correspondence} each maximal ideal of $\alg A_p$ containing $\{x \leq 0\} \cap \Box^{-1} u^{-1}(0)$ has to contain $t_d(a)$. Let $I$ be the $\ell$-ideal generated by the set $\{x \leq 0\} \cap \Box^{-1} u^{-1}(0)$.
    Now $[t_d(a)]_{I} \in \rad(\alg A_p/I)$ and thus by Lemma \ref{l:infinitesimal} we obtain $2^k \cdot [t_d(a)]_{I}-[p]_{I} \geq [0]_I$ for each $k \in \mathbb N$.
    Therefore, there exists $x \in I$ such that $2^k \cdot t_d(a)-p \geq x$ and without loss of generality we can assume $x \leq 0$ (since otherwise we can replace it by $x \land 0 \in I$). From Lemma~\ref{l:almost ideal} we get $x \in \Box^{-1}u^{-1}(0)$ and thus we have by Lemma~\ref{l:modal basics}(ii),(iv) \& (v) that 
     $$0=u(\Box x) \leq u  \big(\Box(2^k \cdot t_d(a)-p)\big) =u \big(\Box (2^k \cdot t_d(a))\big)-u (\Box p)= 2^k \cdot u \big(\Box t_d(a)\big)+1$$ for each $k \in \mathbb N$. 
    Since $-1$ is a negative strong unit in $\R_{-1}$ this means that $u (\Box t_d(a))$ is infinitesimal in $\R_{-1}$. However, since the algebra $\R_{-1}$ is simple this implies $u (\Box t_d(a))=0$, a contradiction.
\end{proof}

The following Truth Lemma is a standard consequence. Note that the set of formulas $\Fm{\Lang_p}$ can be seen as the absolutely free algebra in the language of pointed modal Abelian $\ell$-groups and we identify \emph{algebraic evaluations} $e \colon \Prop \to \tuple{\alg{A}_p, \Box}$ with the corresponding homomorphisms $e \colon \Lang_p \to \tuple{\alg{A}_p, \Box}$.   

\begin{corollary}[\sf Truth Lemma]\label{Lemma:TruthLemma}
    Let $\tuple{\alg A_p, \Box} \in \pmAb$ be strongly negatively pointed, $u \in \Hom(\alg{A}_p,\R_{-1})$ and $e\colon \Prop \to \alg{A}_p$ an algebraic evaluation. Let $\Val_e$ be the propositional valuation on 
    $\mathfrak F_{\tuple{\alg{A}_p,\Box}}$ defined by $\Val_e(u,x)=u(e(x))$.
    Then the property $\Val_e(u,\varphi)=u(e(\varphi))$ extends to every formula $\varphi$.
\end{corollary}
\begin{proof}
Induction on the formula, where the case $\Box\varphi$ follows from Theorem~\ref{theorem:TruthLemmaBox}:  $$
\Val_e(u, \Box \varphi) = \bigwedge_{u R_\Box v} \Val_e(v, \varphi) = \bigwedge_{u R_\Box v} v\big(e(\varphi)\big) = u\big(e(\Box\varphi)\big).
$$
All other cases are clear (since $u$ is a homomorphism with respect to the non-modal language).
\end{proof}  

\subsection{An Infinitary Algebraic Completeness}
\label{subsection:TowardsCompleteness}
Similar to the case of modal $\sf MV$-algebras \cite{HansoulTeheux:ModalLukasiewicz}, the Truth Lemma of the previous section does not immediately yield algebraic completeness. In particular, if the inequation $\varphi \geq 0$ does not hold in some $\tuple{\alg{A}_p, \Box} \in \pmAb$ we have yet to obtain the \emph{witnessing homomorphism} which shows that $\varphi$ is not valid in the corresponding canonical frame $\mathfrak{F}_{\tuple{\alg{A}_p, \Box}}$. To construct such a homomorphism, we will use Lemma~\ref{l:infinitesimal} and introduce an infinitary rule \eqref{InfinitaryRule} based on \eqref{e:infinitesimalNegative} to capture that the corresponding term is not evaluated inside the radical. 

Moreover, we must ensure that a potential counterexample is witnessed in a well-defined model, that is, with a \emph{bounded valuation} (recall Definition~\ref{definition:Model}). We observed in Subsection~\ref{subsection:AlgebrasAndFrames} that this is ensured for canonical frames of \emph{strongly} pointed modal Abelian $\ell$-groups.

To formalize the requirement that it suffices to check $\varphi$ within these bounds, let us first set the following notation. 
For a variable $x_i$ and $k \in \mathbb N$, we define terms 
$$x_i^k := (x_i \lor k \cdot p) \land k \cdot (-p).$$
For a formula $\varphi(x_1,\dots,x_n)$ and $k \in \mathbb N$ we let $\varphi^k$ denote the substitution instance
$$\varphi^k(x_1, \dots, x_n) = \varphi(x_1^k,\dots,x_n^k).$$
Now the \emph{bounded valuations} rule is expressed as follows. 
\begin{equation} \label{r:BE*} \tag{\textsf{BV}}
   \pmAb \models  \varphi^k \geq 0 \text{ for all } k \in \mathbb  N \quad \Longrightarrow \quad \pmAb \models \varphi \geq 0
\end{equation}
Note that the other implication trivially holds.
Let us emphasize that rule \eqref{r:BE*} is \emph{not} a single infinitary rule in the standard sense. Because the syntactic substitutions happen \emph{within} the formula itself, this property cannot be captured by one generalized quasi-equation. Instead, \eqref{r:BE*} represents an infinite family of generalized quasi-equations, one for every formula $\varphi \in \Fm{\Lang_p}$. 

As previously mentioned, we need to combine this with \eqref{e:infinitesimalNegative} to ensure the existence of homomorphisms witnessing non-validity. Thus, we end up with the following.   

\begin{definition}[\sf Infinitary Algebraic Derivability]
\label{definition:InfinitaryDerivability}
We define
\begin{equation}\tag{\textsf{A}$\infty$}\label{InfinitaryRule}
     \vdash_\infty \varphi \quad :\Longleftrightarrow \quad\pmAb \vDash n \cdot (\varphi^k \land 0) - p \geq 0 \text{ for every } n,k \in \mathbb N.
\end{equation}
\end{definition}

Recall that our Truth lemma was only proved for strongly negatively pointed modal Abelian $\ell$-groups. The following lemma will allow us to `move' to such an algebra in the proof of Theorem~\ref{t:completeness} later on.

\begin{lemma} \label{l:soubgroup with strong unit}
    Let $\tuple{\alg A_p,\Box}$ be a negatively pointed modal Abelian $\ell$-group with $p \neq 0$ and let $\alg B$ be the $\ell$-ideal of $\alg A$ generated by $p$. Then $\tuple{\alg B_p,\Box}$ is a strongly negatively pointed modal Abelian $\ell$-group. 
\end{lemma}

\begin{proof}
    We only need to verify that $\alg{B}_p$ is closed under $\Box$. Let $b \in \alg{B}_p$. We have $|b| \leq n \cdot |p|$ and hence we obtain $n \cdot p \leq b \leq -n \cdot p$. By Lemma~\ref{l:modal basics} (iv) we know that the operator $\Box$ is monotone, from which we get (using $-\Box (n \cdot p)=\Box (-n \cdot p)$) that $\Box(n \cdot p) \leq \Box b \leq \Box(-n \cdot p)$ and $\Box(n \cdot p) \leq -\Box b \leq \Box(-n \cdot p)$. Therefore, we showed $|\Box b|\leq |n \cdot p|$ and thus $\Box b  \in \alg{B}_p$.
\end{proof}
We are now ready to prove the main result of this section, completeness of pointed modal Abelian logic with respect to the derivability relation of Definition~\ref{definition:InfinitaryDerivability}. 

\begin{theorem}[\sf Infinitary Algebraic Completeness]\label{t:completeness}
Let $\varphi$ be a formula. Then $\vDash \varphi$ if and only if $\vdash_\infty \varphi$.
\end{theorem}
\begin{proof}
For the implication $(\Rightarrow)$, assume $\nvdash_\infty \varphi$. 
By definition, there exist $n,k \in \mathbb{N}$ such that $ \pmAb \nvDash n \cdot (\varphi^k \land 0)- p \geq 0$. Let $\alg A_p \in \pmAb$ and $e$ be an evaluation on $\alg{A}_p$ such that 
$n \cdot (e(\varphi^k) \land 0) - p \ngeq 0$.
Note that $\alg{A}_p$ satisfies $p \neq0$, since otherwise it would satisfy $e(x_i^k)=(e(x_i) \lor 0) \land 0=0$ for each $x_i$ and thus $n \cdot (e(\varphi^k) \land 0)-p=0$.
Let $\alg B$ be the $\ell$-ideal of $\alg A$ generated by $p$. By Lemma~\ref{l:soubgroup with strong unit} we know that $\alg B_p \in \pmAb$ is strongly negatively pointed.
We define an evaluation $\hat{e}$ on $\alg{B}_p$ by $\hat{e}(x)=e(x^k)$. 
By the above, it holds that $ n \cdot (\hat{e}(\varphi) \land 0) - p \ngeq 0$.
From Lemma~\ref{l:infinitesimal} we obtain that $\hat e(\varphi \land 0)  \notin \rad(\alg B_p)$.
Consequently, there exists a homomorphism $u \in \Hom(\alg B_p,\alg R_{-1})$ such that $u( \hat e(\varphi \land 0))<0$.
Let $\Val_{\hat{e}}$ be the valuation on the canonical frame of $\alg{B}_p$ defined by $\Val_{\hat{e}}(v, x_i) = v(\hat{e}(x_i))$ for all $v \in \Hom(\alg{B}_p, \alg{R}_{-1})$. The Truth Lemma yields 
$$
\Val_{\hat{e}}(u,\varphi)\wedge 0 = \Val_{\hat{e}}(u, \varphi \wedge 0) = u\big(\hat{e}(\varphi \wedge 0)\big) < 0
$$
and therefore clearly $\Val_{\hat{e}}(u, \varphi) < 0$, which witnesses $\not\vDash \varphi$.

For the implication $(\Leftarrow)$, assume $\not\vDash \varphi$ and let $\tuple{W,R,\Val}$ be a model such that $\Val(u,\varphi)<0$ for some $u \in W$. Since $\Val$ is a bounded valuation there is some $k \in \mathbb N$ such that $\Val(u,\varphi^k)<0$ also holds.
Let $e$ be the evaluation on the complex algebra $\mathfrak C_{\tuple{W,R}}$ defined by $e(x_i) = \Val(-,x_i)$. We have $e(\varphi^k)=\Val(-,\varphi^k)$. Since $e(\varphi^k)(u)=\Val(u,\varphi^k)<0$ it follows that $e(\varphi^k \land 0) < 0$.
Since $\{f \in \mathfrak C_{\tuple{W,R}} \colon f(u)= 0\}$ is a maximal $\ell$-ideal of $\mathfrak C_{\tuple{W,R}}$ (because it is the preimage $\mathsf{ev}_u^{-1}(0)$ of the homomorphism $\mathsf{ev}_u \colon \mathfrak{C}_{\tuple{W,R}} \to \alg{R}_{-1}$ evaluating at $u$) and $e(\varphi^k \land 0) \notin \{f \in \mathfrak C_{\tuple{W,R}} \colon f(u)= 0\}$ by Lemma~\ref{l:infinitesimal} it follows there is $n \in \mathbb N$ such that $e(n \cdot (\varphi^k \land 0)- p) \ngeq 0$, which by definition shows $\not\vdash_\infty \varphi$.
This completes the proof.
\end{proof}

\section{Conclusion}
\label{section:Conclusion}
In this paper, we introduced the variety of negatively pointed modal Abelian $\ell$-groups and established its connection to real-valued modal logic, culminating in Theorem~\ref{t:completeness} characterizing modal validity via the rule \eqref{InfinitaryRule}. We take this as evidence that the variety $\pmAb$ deserves further investigation in the future.

Naturally, this also applies to the study of modal Abelian $\ell$-groups \emph{without} a designated point. The crucial role played by the non-zero designated element was highlighted throughout the present work. Without the constant symbol $p$ in the language, every term function $t(x_1, \dots ,x_n)$ over any modal Abelian $\ell$-group trivially satisfies $t(0,\dots,0)=0$, which effectively prevents the construction of the separating terms $t_d$ of Lemma~\ref{l:term separator}. Moreover, the strong unit, the radical, and the boundedness of valuations are all dependent on the assumption of pointedness. Nevertheless, in future work we aim to establish a similar algebraic completeness result for (point-free) real-valued modal logic with respect to the variety $\mAb$.

It is common in many-valued modal logic to not only consider crisp, but \emph{labeled accessibility relations} (e.g.\ see \cite{Bou2009,Busaniche2022} for the case of finitely-valued modal \L ukasiewicz logic). In (pointed) modal Abelian logic, one may consider \emph{bounded} labeled accessibility relations $R\colon W^2 \to [-r,r] \subseteq \mathbb{R}$ with the modality interpreted via $\Val(u,\Box\varphi) = \bigwedge_{v\in W} \Val(v,\varphi) - R(u,v)$. Studying this logic and the relationship to modal \L ukasiewicz logic would clearly be an interesting (albeit challenging) future direction.

There is a well-known categorical equivalence between the class of strongly pointed Abelian $\ell$-groups and the class of $\mathsf{MV}$-algebras established by the \emph{Mundici functor} (e.g.\ see \cite{Cignoli-Ottaviano-Mundici:AlgebraicFoundations}). Similarly, we plan to investigate the category-theoretic relationship between strongly negatively pointed modal Abelian $\ell$-groups and the category $\mathsf{MMV}$ of modal $\mathsf{MV}$-algebras of \cite{HansoulTeheux:ModalLukasiewicz} (let us note here that the Mundici equivalence was extended to \emph{monadic} $\mathsf{MV}$-algebras  in \cite{Cimadamore2011}). As a first step we can show that, given a negatively pointed modal Abelian $\ell$-group, its interval $[p,0]$ is a modal $\mathsf{MV}$-algebra. The construction in the reverse direction appears to be less straightforward. 

Another open question is whether we can replace rule \eqref{InfinitaryRule} with a \emph{single} generalized quasi-equation. As discussed previously, rule \eqref{r:BE*} corresponds to an infinite family of generalized quasi-equations (one for each formula $\varphi$) and consequently defines a generalized quasi-variety, and we want to investigate whether there exists a single generalized quasi-equation defining it. 

Finally, we also want to address questions regarding computational \emph{complexity}. To this end, we plan to generalize and apply the approach presented in \cite{Hanikova-Jankovec:ComplexityUnboundedRelative} for unbounded \L ukasiewicz logic to the setting of strongly pointed modal Abelian $\ell$-groups.

\section*{Acknowledgment}
This work has been funded by a grant from the Programme Johannes Amos Comenius
under the Ministry of Education, Youth and Sports of the Czech Republic,
CZ.02.01.01/00/23\_025/0008711.
The first author was also supported from the project SVV-2025-260837.

\bibliographystyle{eptcs}
\bibliography{mfl}
\end{document}